\documentclass[11pt,a4paper]{article}
\usepackage[utf8]{inputenc}
\usepackage[english]{babel}
\usepackage{amsmath,amssymb,amsfonts,mathrsfs}
\usepackage{physics}
\usepackage{amsthm}
\usepackage{geometry}
\usepackage[colorlinks=true, citecolor=blue, linkcolor=blue, urlcolor=blue]{hyperref}

\newtheorem{theorem}{Theorem}[section]

\newtheorem{proposition}[theorem]{Proposition}

\newtheorem{definition}[theorem]{Definition}
\newtheorem{remark}[theorem]{Remark}

\title{\textbf{An Introduction to the $p$-Differential Forms in Spinor Variables on the Minkowskian Manifold}}
\author{Julio Cesar Jaramillo Quiceno\\Universidad Cooperativa de Colombia\thanks{jcjaramilloq@unal.edu.co}, \thanks{julio.jaramilloq@campusucc.edu.co}}
\date{\today}

\begin{document}

\maketitle

\vspace{0.5cm}

\begin{abstract}
This article presents an introductory study of
$p$-differential forms in spinor variables defined on the Minkowskian (or Lorentzian) manifold. Starting from the algebraic and geometric foundations of spinor spaces and Clifford algebras, we establish the correspondence between Minkowski spacetime and Hermitian matrices in spinor coordinates. Within this framework, we construct differential 1-forms and 2-forms expressed in spinor notation, and formulate the Dirac operator in terms of exterior calculus. Furthermore, we extend this approach to the case of non-monogenic functions, following the method proposed in~\cite{Jaramillo2024}. The results illustrate the consistency of the spinorial exterior calculus and its potential applications in the study of relativistic field equations and geometric formulations of quantum theory.
\newline

\subsubsection*{Keywords:} Spinor variables, Clifford algebra, exterior differential forms, Dirac operator, Minkowski manifold.\newline

\textbf{MSC (2020):} 83Axx, 81T99, 81T20, 58A99

\end{abstract}

\section{Introduction}
In 1964, Bichteler formulated a spinorial exterior differential calculus, in a manner analogous to Cartan’s tensor calculus, which enabled the definition of metric spinors and a more general treatment of spacetime curvature~\cite{Bitcheler1964}. Later, in 1978, Graf proposed a formulation of spinors in terms of differential forms within a Clifford ideal, aiming to avoid the traditional concept of bivalency and to recover the results of the Dirac equation coupled to electromagnetism in flat space~\cite{Graf1978}.
Benn and Tucker subsequently reformulated the Dirac equation using the formalism of exterior calculus and Clifford algebras. They demonstrated that Dirac spinors can be represented in terms of differential forms and established the Lorentz covariance properties of the Dirac operator~\cite{Ben-Tucker1985}.
In parallel, Rodrigues, Souza, and Vaz introduced an alternative framework for defining spinor fields and superfields to describe matter with spin in physical theories, clearly distinguishing their approach from other existing formulations~\cite{RSV1995}.

Lachieze-Rey provided a concise and systematic introduction to spin groups and spinors from a foundational group-theoretic perspective, detailing the construction of low-dimensional Clifford algebras and establishing the intrinsic connection between these two formalisms. His presentation culminates with a discussion of the generalization of these ideas to arbitrary spacetimes through the theory of spin and spinor bundles~\cite{Lachieze2005}.
Torres del Castillo’s work, in turn, offers a rigorous and comprehensive development of the two-component spinor formalism in four-dimensional spaces, emphasizing its intrinsic geometric structure. It thoroughly illustrates the power of this framework for analysing spacetime geometry and for obtaining solutions to key field equations of physical relevance, such as the Dirac and Weyl equations in curved backgrounds~\cite{Torres2010}.

The main objective of the present work is to introduce the concept of \( p \)-differential forms in spinor variables on the Minkowskian (or Lorentzian) manifold. The article is organized as follows:
Section~2 provides the necessary preliminaries, while Section~3 develops the spinor coordinates on the Minkowski manifold. This section also proposes constructions for the differential 1-form and 2-form in spinor coordinates, formulated within the framework of exterior calculus for the Dirac operator and for the operator associated with non-monogenic functions, as proposed in reference~\cite{Jaramillo2024}.

\section{Preliminaries}

\begin{definition}\label{def:0}
\emph{Minkowski space} \(\mathbb{R}^{1,3}\) is the four-dimensional real vector space equipped with the indefinite non-degenerate bilinear form called the Minkowski metric. Formally, it is the pair \((\mathbb{R}^4, \eta)\), where \(\eta : \mathbb{R}^4 \times \mathbb{R}^4 \to \mathbb{R}\) is defined by
\begin{equation}
\eta(x, y) = x^0 y^0 - x^1 y^1 - x^2 y^2 - x^3 y^3,
\end{equation}
for vectors \(x = (x^0, x^1, x^2, x^3)\) and \(y = (y^0, y^1, y^2, y^3)\) in \(\mathbb{R}^4\). In matrix representation, the metric tensor is given by
\begin{equation}
\eta = \mathrm{diag}(+1, -1, -1, -1).
\end{equation}
The coordinate \(x^0 = ct\) represents the time dimension (with \(c\) the speed of light), while \((x^1, x^2, x^3)\) correspond to spatial coordinates. The invariance of the interval
\begin{equation}
ds^2 = \eta_{\mu\nu} dx^\mu dx^\nu = (dx^0)^2 - (dx^1)^2 - (dx^2)^2 - (dx^3)^2
\end{equation}
under Lorentz transformations \(\mathrm{O}(1,3)\) defines the geometry of Minkowski space \cite{Lachieze2005}, \cite{Varlamov2004}.
\end{definition}

\begin{definition}\label{def:1}
The O(1,3) group is the Lorentz group, comprising all linear transformations that preserve the Minkowski space-time metric. Formally, it is the generalized orthogonal group associated with a metric of signature $\mathrm{diag}(1,-1,-1,-1)$ or $\mathrm{diag}(-1,1,1,1)$. It consists of all $4 \times 4$ matrices with real entries that satisfy the relation \cite{Lachieze2005}
\begin{equation}
\Lambda^{T} \eta \Lambda = \eta .
\end{equation}
\end{definition}

\begin{definition}\label{def:2}
The group $\mathrm{SO}(1,3)$ is defined as the proper Lorentz group, consisting of all matrices $\Lambda \in \mathrm{O}(1,3)$ with determinant $+1$, that is,
\begin{equation}
\mathrm{SO}(1,3) := \left\lbrace \Lambda \in \mathrm{O}(1,3) \mid \det(\Lambda) = +1 \right\rbrace.
\end{equation}
This group preserves the space-time orientation (excluding spatial and temporal reflections) but may include transformations that reverse the direction of time (non-orthochronous) \cite{Lachieze2005}.
\end{definition}

\begin{definition}\label{def:3}
The group $\mathrm{SO}^+(1,3)$ is the proper, orthochronous Lorentz group, which is a subgroup of $\mathrm{SO}(1,3)$ that preserves the direction of time. Specifically,
\begin{equation}
\mathrm{SO}^+(1,3) := \left\lbrace \Lambda \in \mathrm{SO}(1,3) \mid \Lambda_{00} \geq 1 \right\rbrace,
\end{equation}
where the condition $\Lambda_{00} \geq 1$ ensures that the transformation does not reverse the time direction \cite{Lachieze2005}.
\end{definition}

\begin{definition}\label{def:4}
The Clifford algebra \(\text{Cl}(\mathbb{R}^{1,3})\) is the associative algebra generated by the vector space \(\mathbb{R}^{1,3}\) equipped with the Minkowski metric \(\eta = \text{diag}(+1, -1, -1, -1)\), subject to the fundamental relation:
\begin{equation}
\gamma_i \gamma_j + \gamma_j \gamma_i = 2\eta_{ij}\mathbf{1}, \quad \text{for } i,j = 0,1,2,3,
\end{equation}
where \(\{\gamma_0, \gamma_1, \gamma_2, \gamma_3\}\) are the generators of the algebra, and \(\mathbf{1}\) denotes the identity element. Equivalently, for any vectors \(u, v \in \mathbb{R}^{1,3}\), the Clifford product satisfies:
\begin{equation}
uv + vu = 2\eta(u,v)\mathbf{1}.
\end{equation}
As a real vector space, \(\text{Cl}(\mathbb{R}^{1,3})\) has dimension \(2^4 = 16\), and it is isomorphic to \(\text{Mat}(4, \mathbb{R})\) as an algebra \cite{Lachieze2005}.
\end{definition}

\begin{definition}\label{def:5}
The Clifford algebra \(\text{Cl}(\mathbb{R}^{1,3})\) is the quotient of the tensor algebra \(T(\mathbb{R}^{1,3})\) by the ideal generated by elements of the form \(v \otimes v - \eta(v,v)\mathbf{1}\) for all \(v \in \mathbb{R}^{1,3}\) \cite{Lachieze2005}, \cite{Varlamov2004}.
\end{definition}

\begin{definition}\label{def:6}
The group $\mathrm{Mat}_2(\mathbb{C})$ consists of all complex $2 \times 2$ matrices \cite{Lachieze2005}:
\begin{equation}
\mathrm{Mat}_2(\mathbb{C}) := \left\lbrace
\begin{bmatrix}
a & b \\
c & d
\end{bmatrix}
\ \middle|\ a, b, c, d \in \mathbb{C} \right\rbrace.
\end{equation}
\end{definition}

\begin{definition}\label{def:7}
The group $\mathrm{SL}(2,\mathbb{C})$ is the complex special linear group consisting of all $2 \times 2$ complex matrices with determinant $+1$, that is \cite{Lachieze2005},
\begin{equation}
\mathrm{SL}(2,\mathbb{C}) := \left\lbrace M \in \mathrm{Mat}_2(\mathbb{C}) \mid \det M = +1 \right\rbrace.
\end{equation}
\end{definition}

\begin{definition}\label{def:8}
The set $\mathrm{Herm}(2)$ consists of all $2 \times 2$ Hermitian matrices \cite{Lachieze2005}
\begin{equation}
\mathrm{Herm}(2) := \left\lbrace
H \in \mathrm{Mat}_2(\mathbb{C}) \ \middle|\ H^\dagger = H
\right\rbrace,
\end{equation}
where $H^\dagger$ denotes the conjugate transpose of $H$. Explicitly, any $H \in \mathrm{Herm}(2)$ can be written as
\begin{equation}
H = \begin{bmatrix}
t + z & x - iy \\
x + iy & t - z
\end{bmatrix},
\end{equation}
with $t, x, y, z \in \mathbb{R}$.
\end{definition}

\begin{definition}\label{def:9}
The group $\mathrm{Spin}(1,3)$ is the double connected cover of the proper orthochronous Lorentz group $\mathrm{SO}^+(1,3)$ and is a subgroup of the Clifford algebra $\mathrm{Cl}(\mathbb{R}^{1,3})$ associated with Minkowski space $\mathbb{R}^{1,3}$. This group consists of elements $S \in \mathrm{Cl}(\mathbb{R}^{1,3})$ satisfying \cite{Lachieze2005}, \cite{Varlamov2004}
\begin{equation}
S \gamma^{\mu} S^{-1} = \Lambda^{\mu}_{\nu} \gamma^{\nu},
\end{equation}
where $\Lambda \in \mathrm{SO}^+(1,3)$. Furthermore, $\mathrm{Spin}(1,3)$ is isomorphic to the group $\mathrm{SL}(2, \mathbb{C})$.
\end{definition}

\begin{remark}\label{rem:10}
The \emph{Minkowski Space-time:} is prevalent in physics, particularly in the context of special relativity. It emphasizes the \textit{affine} and \textit{vector space} structure. One can add points (events) with coordinates and there is a privileged origin. The primary focus is on the global $(t, x, y, z)$ coordinate system and the invariant interval $ds^2 = \eta_{\mu\nu}dx^\mu dx^\nu$ \cite{Misner-Thorne-Wheeler1973}.
\end{remark}

\begin{definition}\label{def:11}
The \emph{Minkowski manifold} is the quadruple $(\mathbb{R}^{1,3}, \eta, \mathrm{SO}^+(1,3), \mathrm{SL}(2, \mathbb{C}))$, where the metric $\eta$ defines the causal geometry, the group $\mathrm{SO}^+(1,3)$ describes its isometric symmetries, and the group $\mathrm{SL}(2, \mathbb{C})$ acts as its double cover, providing the spin structure and a fundamental representation via the matrices of $\mathrm{Herm}(2)$ \cite{Lachieze2005}.
\end{definition}

\begin{definition}\label{def:12}
The \emph{Dirac matrices in Minkowski space} are defined in the Weyl (chiral) representation as \cite{Schmidt-Watcher2007}
\begin{equation}\label{eq:min-gamma}
\gamma^\mu_{ab} = \begin{bmatrix}
0 & \sigma^\mu_{\alpha\dot{\beta}} \\
\bar{\sigma}^{\mu}_{\dot{\alpha}\beta} & 0
\end{bmatrix},
\end{equation}
where $\sigma^\mu = (\mathbf{1}_2, \sigma^i)$ and $\bar{\sigma}^\mu = (\mathbf{1}_2, -\sigma^i)$ denoting the Pauli matrices in the Minkowskian space. The \emph{spinor indices} $\alpha, \beta$ run over $\{1,2\}$ for left-handed Weyl spinors, while $\dot{\alpha}, \dot{\beta}$ run over $\{\dot{1}, \dot{2}\}$ for right-handed Weyl spinors.
\end{definition}

\begin{definition}\label{def:13}
Let $\mathcal{M} = \mathbb{R}^{1,3}$ be Minkowski space with coordinates $(x^0, x^1, x^2, x^3)$ and metric $\eta = \mathrm{diag}(+1, -1, -1, -1)$ \cite{Lachieze2005}.

\begin{enumerate}
    \item A \textbf{1-form} (or differential 1-form) is a section of the cotangent bundle $T^*M$, which can be expressed locally as \cite{Misner-Thorne-Wheeler1973}
    \begin{equation}
    \omega = \omega_\mu dx^\mu = \omega_0 dx^0 + \omega_1 dx^1 + \omega_2 dx^2 + \omega_3 dx^3,
    \end{equation}
    where $\omega_\mu$ are smooth functions and $\{dx^\mu\}$ form a basis for the cotangent space.

    \item A \textbf{2-form} is an antisymmetric $(0,2)$-tensor that can be written as \cite{Misner-Thorne-Wheeler1973}
    \begin{equation}
    F = \frac{1}{2} F_{\mu\nu} dx^\mu \wedge dx^\nu,
    \end{equation}
    where $F_{\mu\nu} = -F_{\nu\mu}$ are smooth functions, and the wedge product satisfies $dx^\mu \wedge dx^\nu = -dx^\nu \wedge dx^\mu$.
\end{enumerate}

The exterior derivative acts on these forms as:
\begin{itemize}
    \item For a 0-form (function) $f$: $df = \partial_\mu f dx^\mu$
    \item For a 1-form $\omega$: $d\omega = \partial_\mu \omega_\nu dx^\mu \wedge dx^\nu$
\end{itemize}

The Hodge star operator $\star$, which depends on the metric $\eta$, maps $p$-forms to $(4-p)$-forms in four dimensions.
\end{definition}

\section{Exterior calculus in the spinor coordinates}

\subsection{Spinor coordinates in the Minkowski manifold}

\begin{definition}
Spinor coordinates establish a correspondence between Minkowski space-time and the set of Hermitian matrices $\mathrm{Herm}(2) \subset \mathrm{Mat}_2(\mathbb{C})$. For coordinates $x^{\mu} \in M$, this correspondence is given by \cite{Lachieze2005}, \cite{Torres2010}
\begin{equation}
X = x^{\mu} \sigma_{\mu} = \begin{bmatrix}
x^{1\dot{1}} & x^{1\dot{2}} \\
x^{2\dot{1}} & x^{2\dot{2}}
\end{bmatrix} =
\begin{bmatrix}
x^0 + x^3 & x^1 - i x^2 \\
x^1 + i x^2 & x^0 - x^3
\end{bmatrix},
\end{equation}
where $\sigma_\mu = (\mathbf{1}, \vec{\sigma})$ with $\mu = 0,1,2,3$.
\end{definition}

\begin{remark}
Regarding notation, the use of indices in Minkowski space-time (Minkowski manifold) for spinors can be substituted by a change of indices for a standard tensor. For example, a vector $x^{\mu}$ can be expressed in spinor notation as $x^{\alpha\dot{\alpha}}$, and similarly, the derivative operator $\partial_{\mu}$ can be written as $\partial_{\alpha\dot{\alpha}}$ \cite{Lachieze2005}, \cite{Torres2010}.
\end{remark}

\begin{definition}
Spinor coordinates can be explicitly defined as follows \cite{Lachieze2005}, \cite{Torres2010}:
\begin{align}\label{eq:coordinates}
x_{\alpha\dot{\beta}}  & \equiv \ \frac{1}{\sqrt{2}} \sigma_{\alpha\dot{\beta}}^{\mu} x_{\mu}, \\
\label{eq:coordinates1}\partial^{\alpha\dot{\beta}} & \equiv \ \frac{1}{\sqrt{2}} \sigma^{\alpha\dot{\beta}}_{\nu} \partial^{\nu},
\end{align}
for $\mu = 0,1,2,3$.
\end{definition}

Following the above developments, we now define differential 1-forms and 2-forms in spinor notation.

\begin{proposition}
According to Eq.~\eqref{eq:coordinates} and Definition~\ref{def:13}, we can express the differential 1-form and 2-form in spinor notation as follows \cite{Torres2010}
\begin{align}\label{eq:relacion1}
dx_{\alpha\dot{\beta}} & \equiv \ \frac{1}{\sqrt{2}} \sigma_{\alpha\dot{\beta}}^{\mu} dx_{\mu}, \\ \label{eq:relacion2}
dx_{\alpha\dot{\beta}} & = \ i\epsilon^{\mu\nu\lambda}dx_{\mu}\wedge dx_{\nu},
\end{align}
where $\epsilon^{\mu\nu\lambda}$ is the Levi-Civita tensor.
\end{proposition}

\begin{proof}
For the second relation, we express the spinorial differential form as:
\begin{equation}\label{eq:dx}
dx_{\alpha\dot{\beta}} = \frac{1}{2}(\sigma_{\alpha\dot{\beta}}^{\mu} \bar{\sigma}^{\nu\dot{\beta}\alpha} - \sigma_{\alpha\dot{\beta}}^{\nu} \bar{\sigma}^{\mu\dot{\beta}\alpha}) dx_{\mu} \wedge dx_{\nu}.
\end{equation}

Using the Pauli matrix identity:
\begin{equation}\label{eq:spin}
\sigma_{\alpha\dot{\beta}}^{\mu} \bar{\sigma}^{\nu\dot{\beta}\alpha} - \sigma_{\alpha\dot{\beta}}^{\nu} \bar{\sigma}^{\mu\dot{\beta}\alpha} = 2i \epsilon^{\mu\nu\lambda\rho} \sigma_{\lambda\alpha\dot{\beta}},
\end{equation}
and substituting \eqref{eq:spin} into \eqref{eq:dx}, we obtain the relation \eqref{eq:relacion2}.
\end{proof}

\begin{proposition}
For the case of Dirac matrices, we employ a relation analogous to \eqref{eq:spin} which corresponds to the spin tensor $\sigma^{\mu\nu}_{\alpha\dot{\beta}}$, that is
\begin{equation}\label{eq:tensor-espinorial}
\gamma_{ab}^{\mu}\gamma_{ab}^{\nu}-\gamma_{ab}^{\nu}\gamma_{ab}^{\mu}=2i\sigma_{\alpha\dot{\beta}}^{\mu\nu},
\end{equation}
where  $\gamma_{ab}^{\mu}$ is given by (\ref{eq:min-gamma}).
\end{proposition}

In view of the preceding discussion, we define a spinor field $\psi$ on the Minkowski manifold $\mathcal{M}$ in the following manner:

\begin{definition}\label{def:function}
Let $\mathcal{M}$ denote the Minkowski manifold. A spinor field $\psi$ on $\mathcal{M}$ is defined as a mapping
\begin{align*}
\psi : \mathcal{M} &\longrightarrow \mathbb{C}^4, \\
x_{\mu} &\longmapsto \psi(x_{\mu}),
\end{align*}
where $x_{\mu} = (x^0, x^1, x^2, x^3)$ represents space-time coordinates, and the field $\psi(x)$ carries spinor structure that transforms appropriately under Lorentz transformations.
\end{definition}

\begin{proposition}\label{def:Dirac-operator}
The Dirac operator in spinor notation is given by:
\begin{equation}\label{eq:Dirac}
D_{\alpha\dot{\beta}}=\sqrt{2}(\Pi^{+}\otimes\partial_{\alpha\dot{\beta}}+\Pi^{-}\otimes\partial_{\dot{\alpha}\beta}),
\end{equation}
where $\Pi^{+}$ and $\Pi^{-}$ are matrices $2\times 2$:
\begin{equation}
\Pi^{+}=\begin{bmatrix}
0 & 1\\
0 & 0
\end{bmatrix}, \quad \Pi^{-}=\begin{bmatrix}
0 & 0\\
1 & 0
\end{bmatrix}.
\end{equation}
\end{proposition}

\begin{proof}
First, we can rewrite Eq. (\ref{eq:coordinates1}) in the following form:
\begin{equation}\label{eq:coordinates2}
\sqrt{2}\partial^{\alpha\dot{\beta}} = \sigma^{\alpha\dot{\beta}}_{\mu}\partial^{\mu}.
\end{equation}

Multiplying Eq. (\ref{eq:min-gamma}) by $\partial_{\mu}$ yields the Dirac operator:
\begin{equation}\label{eq:coordinates3}
\partial_{\mu}\gamma^{\mu}_{ab} =
\begin{bmatrix}
0 & \partial_{\mu}\sigma^\mu_{\alpha\dot{\beta}} \\
\partial_{\mu}\bar{\sigma}^{\mu}_{\dot{\alpha}\beta} & 0
\end{bmatrix}.
\end{equation}

This can, in turn, be expressed as a linear combination of matrices:
\begin{equation}\label{eq:linear}
\partial_\mu\gamma^{\mu}_{ab} = \begin{bmatrix}
0 & 1\\
0 & 0
\end{bmatrix} \otimes \left(\sigma^{\mu}_{\alpha\dot{\beta}}\partial_\mu\right) + \begin{bmatrix}
0 & 0\\
1 & 0
\end{bmatrix} \otimes \left(\bar{\sigma}^{\mu}_{\dot{\alpha}\beta}\partial_{\mu}\right).
\end{equation}

Finally, substituting Eq. (\ref{eq:coordinates2}) into Eq. (\ref{eq:linear}), we obtain:
\begin{equation}\label{eq:final}
\partial_\mu\gamma^{\mu}_{ab} = \sqrt{2} \begin{bmatrix}
0 & 1\\
0 & 0
\end{bmatrix} \otimes \partial_{\alpha\dot{\beta}} + \sqrt{2} \begin{bmatrix}
0 & 0\\
1 & 0
\end{bmatrix} \otimes \partial_{\dot{\alpha}\beta},
\end{equation}
and denoting $\begin{bmatrix} 0 & 1 \\ 0 & 0 \end{bmatrix}$ and $\begin{bmatrix} 0 & 0 \\ 1 & 0 \end{bmatrix}$ as $\Pi^{+}$ and $\Pi^{-}$ respectively, we finally obtain (\ref{eq:Dirac}); with this, we conclude the proof.
\end{proof}

\subsection{Differential forms for the spinor coordinates}

\begin{proposition}\label{prop:Tensor-Dirac}
From Definition (\ref{def:function}) and Eq. (\ref{eq:tensor-espinorial}) and the above definition, we construct a differential 2-form for the spinor field $\psi$ expressed through the Dirac operator in spinor notation. This is
\begin{equation}\label{eq:2-form-diff}
d\psi =i\sigma^{\alpha\dot{\beta}}_{\mu\nu}D_{\alpha\dot{\beta}}\psi\wedge dx^{\mu}\wedge dx^{\nu}.
\end{equation}
\end{proposition}

\begin{proof}
Let us first consider the analogue of the 2-form in Eq. (\ref{eq:dx}) in relation to Eq. (\ref{eq:tensor-espinorial}), that is:
\begin{equation}\label{eq:2-cliffford}
dx_{\alpha\dot{\beta}} = \frac{i}{2}(\gamma_{ab}^{\mu}\gamma_{ab}^{\nu}-\gamma_{ab}^{\nu}\gamma_{ab}^{\mu})dx_\mu \wedge dx_\nu =i\sigma^{\mu\nu}_{\alpha\dot{\beta}} dx^{\mu} \wedge dx^{\nu},
\end{equation}
Then
\begin{align*}
d\psi  & = \ D_{\alpha\dot{\beta}}\psi\wedge dx_{\alpha\dot{\beta}},\\
& = \ D_{\alpha\dot{\beta}}\psi\wedge i\sigma^{\mu\nu}_{\alpha\dot{\beta}}dx_{\mu}\wedge dx_{\nu},\\
& = \ i\sigma^{\mu\nu}_{\alpha\dot{\beta}}D_{\alpha\dot{\beta}}\psi\wedge dx_{\mu}\wedge dx_{\nu}.
\end{align*}
\end{proof}

\begin{remark}
The Dirac operator in spinor notation applied to function $\psi$, is given by $D_{\alpha\dot{\beta}}\psi$.
\end{remark}

Now, we can apply the same treatment to the case of non-monogenic functions as proposed in reference \cite{Jaramillo2024}.

\begin{definition}
The expression (17) from reference \cite{Jaramillo2024} can be formulated in spinor notation as follows:
\begin{equation}\label{eq:dirac-spinor}
\textit{\textbf{D}}_{\alpha\dot{\beta}} = \gamma^{\mu}_{\alpha\dot{\beta}} \partial_{\nu} + \gamma^{\nu}_{\alpha\dot{\beta}} \partial_{\mu}.
\end{equation}
Here, $\gamma^{\mu}_{\alpha\dot{\beta}}$ and $\gamma^{\nu}_{\alpha\dot{\beta}}$ are spinor matrices defined by
\begin{equation}\label{eq:gamma-definitions}
\gamma^{\mu}_{\dot{\alpha}\beta} = \frac{1}{4} \gamma^{\mu}_{ab} \gamma^{\nu}_{ab} \gamma^{\rho}_{ab} \gamma_{\mu}^{ab} \eta_{\mu\nu}, \quad
\gamma^{\nu}_{ab} = \frac{1}{4} \gamma^{\nu}_{ab} \gamma^{\mu}_{ab} \gamma^{\rho}_{ab} \gamma_{\nu}^{ab} \eta_{\mu\nu},
\end{equation}
which satisfy the algebraic relations specified in equations (14) and (15) of reference \cite{Jaramillo2024}.
\end{definition}

\begin{proposition}\label{prop:spinor-forms}
Using the preceding propositions and following an approach analogous to the Dirac operator studied in Proposition 3.8, we can define differential 1-forms and 2-forms in spinor notation as follows.

The differential 1-form in spinor notation is expressed as:
\begin{equation}\label{eq:1-form}
x^{\alpha\dot{\beta}} = x^{\mu} dx^{\mu} + x^{\nu} dx^{\nu} + x^{\mu} dx^{\nu} \gamma^{\nu}_{\alpha\dot{\beta}} \gamma^{\mu}_{\alpha\dot{\beta}} + x^{\nu} dx^{\mu} \gamma^{\mu}_{\alpha\dot{\beta}} \gamma^{\nu}_{\alpha\dot{\beta}},
\end{equation}

and the corresponding differential 2-form is given by:
\begin{equation}\label{eq:2-form}
dx^{\alpha\dot{\beta}} = x^{\mu} dx^{\mu} \wedge dx^{\nu} + x^{\nu} dx^{\nu} \wedge dx^{\mu} + x^{\mu} dx^{\nu} \gamma^{\nu}_{\alpha\dot{\beta}} \gamma^{\mu}_{\alpha\dot{\beta}} \wedge dx^{\mu} + x^{\nu} dx^{\mu} \gamma^{\mu}_{\alpha\dot{\beta}} \gamma^{\nu}_{\alpha\dot{\beta}} \wedge dx^{\nu}.
\end{equation}
\end{proposition}

\begin{proof}
We begin by expanding the expression for the spinor-valued 1-form:
\begin{align*}
x^{\alpha\dot{\beta}} & = (x^{\mu}\gamma^{\nu}_{\alpha\dot{\beta}} + x^{\nu}\gamma^{\mu}_{\alpha\dot{\beta}})(dx^{\mu}\gamma^{\nu}_{\alpha\dot{\beta}} + dx^{\nu}\gamma^{\mu}_{\alpha\dot{\beta}}),\\
& = x^{\mu}dx^{\mu}(\gamma^{\nu}_{\alpha\dot{\beta}})^{2} + x^{\mu}dx^{\nu}\gamma^{\nu}_{\alpha\dot{\beta}}\gamma^{\mu}_{\alpha\dot{\beta}} + x^{\nu}dx^{\mu}\gamma^{\mu}_{\alpha\dot{\beta}}\gamma^{\nu}_{\alpha\dot{\beta}} + x^{\nu}dx^{\nu}(\gamma^{\mu}_{\alpha\dot{\beta}})^{2}.
\end{align*}

Taking into account the  algebra property $(\gamma^{\mu}_{\alpha\dot{\beta}})^{2} = (\gamma^{\nu}_{\alpha\dot{\beta}})^{2} = \mathbb{I}_{4}$, we obtain equation (\ref{eq:1-form}). The 2-form in equation (\ref{eq:2-form}) is defined by hypothesis, following the standard construction of differential forms from the exterior derivative of the 1-form.
\end{proof}

Taking into consideration Definition 3.6 and 3.10, we can define the differential form for the field $\psi$ if and only if this field is non-monogenic \cite{Jaramillo2024}.

\begin{theorem}
Using expressions (\ref{eq:1-form}) and (\ref{eq:2-form}) together with Definitions 3.6 and 3.10, the exterior derivative of $\psi$ takes the form
\begin{multline}
d\psi  = \textit{\textbf{D}}_{\alpha\dot{\beta}}\psi\wedge x^{\mu}dx^{\mu}\wedge dx^{\nu}+\textit{\textbf{D}}_{\alpha\dot{\beta}}\psi\wedge x^{\nu}dx^{\nu}\wedge dx^{\mu}+
\gamma^{\mu}_{\alpha\dot{\beta}}\partial_{\mu}\psi\wedge x^{\nu}dx^{\mu}\wedge dx^{\nu}\\+\gamma^{\nu}_{\alpha\dot{\beta}}\partial_{\nu}\psi\wedge x^{\nu}dx^{\mu}\wedge dx^{\nu}+ \gamma^{\mu}_{\alpha\dot{\beta}}\gamma^{\nu}_{\alpha\dot{\beta}}\gamma^{\mu}_{\alpha\dot{\beta}}\partial_{\nu}\psi\wedge x^{\mu}dx^{\nu}\wedge dx^{\mu}+\gamma^{\nu}_{\alpha\dot{\beta}}\gamma^{\mu}_{\alpha\dot{\beta}}\gamma^{\nu}_{\alpha\dot{\beta}}\partial_{\mu}\psi\wedge x^{\nu}dx^{\mu}\wedge dx^{\nu}.
\end{multline}
\end{theorem}

\begin{proof}
First, we compute $\textit{\textbf{D}}_{\alpha\dot{\beta}}\psi \wedge dx^{\alpha\dot{\beta}}$
\[
\begin{aligned}
&\bigl(\gamma^{\mu}_{\alpha \dot{\beta}} \, \partial_{\nu}
 + \gamma^{\nu}_{\alpha \dot{\beta}} \, \partial_{\mu}\bigr)
 \wedge
 \bigl(x^{\mu} d x^{\mu} \wedge d x^{\nu}
 + x^{\nu} d x^{\nu} \wedge d x^{\mu}
 + x^{\mu} d x^{\nu} \, \gamma^{\nu}_{\alpha \dot{\beta}}\gamma^{\mu}_{\alpha\dot{\beta}} \wedge d x^{\mu}
 + x^{\nu}dx^{\mu}\gamma^{\mu}_{\alpha\dot{\beta}}\gamma^{\nu}_{\alpha\dot{\beta}}\wedge dx^{\nu}\bigr) \\[1.2em]
&= \,
\gamma^{\mu}_{\alpha \dot{\beta}} \, \partial_{\nu} \psi
  \wedge x^{\mu} d x^{\mu} \wedge d x^{\nu}
+ \gamma^{\mu}_{\alpha \dot{\beta}} \, \partial_{\nu} \psi
  \wedge x^{\nu} d x^{\nu} \wedge d x^{\mu} \\[0.4em]
&\quad
+ \gamma^{\mu}_{\alpha \dot{\beta}} \, \partial_{\nu} \psi
  \wedge x^{\mu} d x^{\nu}\,\gamma^{\nu}_{\alpha\dot{\beta}}\gamma^{\mu}_{\alpha\dot{\beta}} \wedge d x^{\mu}
+ \gamma^{\mu}_{\alpha\dot{\beta}}\partial_{\nu}\psi\wedge x^{\nu}dx^{\mu}\gamma^{\mu}_{\alpha\dot{\beta}}\gamma^{\nu}_{\alpha\dot{\beta}}\wedge dx^{\nu} \\[1.2em]
&\quad
+ \gamma^{\nu}_{\alpha \dot{\beta}} \, \partial_{\mu} \psi
  \wedge x^{\mu} d x^{\mu} \wedge d x^{\nu}
+ \gamma^{\nu}_{\alpha \dot{\beta}} \, \partial_{\mu} \psi
  \wedge x^{\nu} d x^{\nu} \wedge d x^{\mu} \\[0.4em]
&\quad
+ \gamma^{\nu}_{\alpha \dot{\beta}} \, \partial_{\mu} \psi
  \wedge x^{\mu} d x^{\nu}\,\gamma^{\nu}_{\alpha\dot{\beta}}\gamma^{\mu}_{\alpha\dot{\beta}} \wedge d x^{\nu}
+ \gamma^{\nu}_{\alpha\dot{\beta}}\partial_{\mu}\psi\wedge x^{\nu}dx^{\mu}\gamma^{\mu}_{\alpha\dot{\beta}}\gamma^{\nu}_{\alpha\dot{\beta}}\wedge dx^{\nu},
\end{aligned}
\]
therefore
\[
\begin{aligned}
&\bigl(\gamma^{\mu}_{\alpha \dot{\beta}} \, \partial_{\nu} \psi
+
\gamma^{\nu}_{\alpha \dot{\beta}} \, \partial_{\mu} \psi\bigr)
\wedge
x^{\mu} d x^{\mu} \wedge d x^{\nu}
+\bigl(\gamma^{\mu}_{\alpha \dot{\beta}} \, \partial_{\nu} \psi
+
\gamma^{\nu}_{\alpha \dot{\beta}} \, \partial_{\mu} \psi\bigr)
\wedge
x^{\nu} d x^{\nu} \wedge d x^{\mu} \\[1.2em]
&\quad+
\bigl(\gamma^{\mu}_{\alpha \dot{\beta}} \, \partial_{\nu} \psi
+
\gamma^{\nu}_{\alpha \dot{\beta}} \, \partial_{\mu} \psi\bigr)
\wedge
x^{\mu} d x^{\nu}\,
\gamma^{\nu}_{\alpha \dot{\beta}} \gamma^{\mu}_{\alpha\dot{\beta}}
\wedge d x^{\mu}
+\bigl(\gamma^{\mu}_{\alpha \dot{\beta}} \, \partial_{\nu} \psi
+
\gamma^{\nu}_{\alpha \dot{\beta}} \, \partial_{\mu} \psi\bigr)
\wedge
x^{\nu} d x^{\mu} \,
\gamma^{\mu}_{\alpha \dot{\beta}} \gamma^{\nu}_{\alpha \dot{\beta}}
\wedge d x^{\nu},
\end{aligned}
\]
thus
\begin{align*}
\textit{\textbf{D}}_{\alpha \dot{\beta}} \psi \wedge x^{\mu} d x^{\mu} \wedge d x^{\nu}
+ \textit{\textbf{D}}_{\alpha \dot{\beta}} \psi \wedge x^{\nu} d x^{\nu} \wedge d x^{\mu} \\[0.6em]
+ \gamma^{\mu}_{\alpha \dot{\beta}} \gamma^{\nu}_{\alpha \dot{\beta}} \gamma^{\mu}_{\alpha \dot{\beta}}
\, \partial_{\nu} \psi \wedge x^{\mu} d x^{\nu} \wedge d x^{\mu}+ \gamma^{\nu}_{\alpha \dot{\beta}} \gamma^{\nu}_{\alpha \dot{\beta}} \gamma^{\mu}_{\alpha \dot{\beta}}
\, \partial_{\mu} \psi \wedge x^{\nu} d x^{\mu} \wedge d x^{\nu}\\[0.6em]
+ \gamma^{\mu}_{\alpha \dot{\beta}} \gamma^{\mu}_{\alpha \dot{\beta}} \gamma^{\nu}_{\alpha \dot{\beta}}
\, \partial_{\nu} \psi \wedge x^{\nu} d x^{\mu} \wedge d x^{\nu}+ \gamma^{\nu}_{\alpha \dot{\beta}} \gamma^{\mu}_{\alpha \dot{\beta}} \gamma^{\nu}_{\alpha \dot{\beta}}
\, \partial_{\mu} \psi \wedge x^{\nu} d x^{\mu} \wedge d x^{\nu}.
\end{align*}

Taking into account the expression (\ref{eq:dirac-spinor}) we get

\begin{align*}
\textit{\textbf{D}}_{\alpha\dot{\beta}}\psi\wedge x^{\mu}dx^{\mu}\wedge dx^{\nu}+\textit{\textbf{D}}_{\alpha\dot{\beta}}\psi\wedge x^{\nu}dx^{\nu}\wedge dx^{\mu} \,
&+ \gamma^{\mu}_{\alpha \dot{\beta}} \gamma^{\nu}_{\alpha \dot{\beta}}
   \gamma^{\mu}_{\alpha \dot{\beta}} \, \partial_{\nu} \psi
   \wedge x^{\mu} d x^{\nu} \wedge d x^{\mu}\\
+ \gamma^{\nu}_{\alpha \dot{\beta}} \gamma^{\nu}_{\alpha \dot{\beta}}
   \gamma^{\mu}_{\alpha \dot{\beta}} \, \partial_{\mu} \psi
   \wedge x^{\nu} d x^{\mu} \wedge d x^{\nu}
& + \gamma^{\mu}_{\alpha \dot{\beta}} \gamma^{\mu}_{\alpha \dot{\beta}}
   \gamma^{\nu}_{\alpha \dot{\beta}} \, \partial_{\nu} \psi
   \wedge x^{\nu} d x^{\mu} \wedge d x^{\nu}\\
& + \gamma^{\nu}_{\alpha \dot{\beta}} \gamma^{\mu}_{\alpha \dot{\beta}}
   \gamma^{\nu}_{\alpha \dot{\beta}} \, \partial_{\mu} \psi
   \wedge x^{\nu} d x^{\mu} \wedge d x^{\nu},
\end{align*}

finally

\begin{align*}
\textit{\textbf{D}}_{\alpha\dot{\beta}}\psi\wedge x^{\mu}dx^{\mu}\wedge dx^{\nu}+\textit{\textbf{D}}_{\alpha\dot{\beta}}\psi\wedge x^{\nu}dx^{\nu}\wedge dx^{\mu} +\gamma^{\mu}_{\alpha\dot{\beta}}\partial_{\mu}\psi\wedge x^{\nu}dx^{\mu}\wedge dx^{\nu}+\\
\gamma^{\nu}_{\alpha\dot{\beta}}\partial_{\nu}\psi\wedge x^{\nu}dx^{\mu}\wedge dx^{\nu}+\gamma^{\mu}_{\alpha\dot{\beta}}\gamma^{\nu}_{\alpha\dot{\beta}}\gamma^{\mu}_{\alpha\dot{\beta}}\partial_{\nu}\psi\wedge x^{\mu}dx^{\nu}\wedge dx^{\mu}+\\
\gamma^{\nu}_{\alpha\dot{\beta}}\gamma^{\mu}_{\alpha\dot{\beta}}\gamma^{\nu}_{\alpha\dot{\beta}}\partial_{\mu}\psi\wedge x^{\nu}dx^{\mu}\wedge dx^{\nu}.
\end{align*}
\end{proof}

\end{document}